\documentclass[journal]{IEEEtran}

\usepackage{amsmath,amssymb,amsthm}
\usepackage{mathtools}
\usepackage{booktabs}
\usepackage{cite}

\newtheorem{theorem}{Theorem}
\newtheorem{lemma}[theorem]{Lemma}
\newtheorem{corollary}[theorem]{Corollary}

\theoremstyle{definition}
\newtheorem{definition}[theorem]{Definition}
\newtheorem{example}[theorem]{Example}
\theoremstyle{remark}
\newtheorem{remark}[theorem]{Remark}

\newcommand{\Rd}{\mathbb{R}^d}
\newcommand{\R}{\mathbb{R}}
\newcommand{\E}{\mathbb{E}}
\newcommand{\1}{\mathbf{1}}
\newcommand{\Hent}{H}
\newcommand{\abs}[1]{\lvert #1 \rvert}
\newcommand{\norm}[1]{\lVert #1 \rVert}
\newcommand{\dx}{\,dx}
\DeclareMathOperator*{\esssup}{ess\,sup}

\begin{document}

\title{Convergence of Differential Entropies---{II}}

\author{Mahesh~Godavarti%
\thanks{Manuscript submitted \today.  This paper builds
on~\cite{GodavartiHero2004}.}}

\maketitle

\begin{abstract}
We show that under convergence in measure of probability density
functions, differential entropy converges whenever the entropy
integrands $f_n \abs{\log f_n}$ are uniformly integrable and
tight---a direct consequence of Vitali's convergence theorem.
We give an entropy-weighted Orlicz condition:
$\sup_n \int f_n\, \Psi(\abs{\log f_n}) < \infty$ for a single
superlinear~$\Psi$, strictly weaker than the fixed-$\alpha$
condition of Godavarti and Hero~(2004).  We also disprove the
Godavarti--Hero conjecture that $\alpha > 1$ could be replaced by
$\alpha_n \downarrow 1$.  We recover the sufficient conditions of
Godavarti--Hero, Piera--Parada, and Ghourchian--Gohari--Amini
as corollaries.  On bounded domains, we prove that uniform
integrability of the entropy integrands is both necessary and
sufficient---a complete characterization of entropy convergence.
\end{abstract}

\begin{IEEEkeywords}
Differential entropy, convergence, uniform integrability, Vitali
convergence theorem, convergence in measure.
\end{IEEEkeywords}

\section{Introduction}\label{sec:intro}

If a sequence of probability density functions $\{f_n\}$ converges to
a limiting density~$f$, when does the differential entropy converge?
That is, under what additional conditions does
\begin{equation}\label{eq:entropy-conv}
  f_n \to f
  \quad\Longrightarrow\quad
  \Hent(f_n) \to \Hent(f),
\end{equation}
where $\Hent(f) = -\int_{\Rd} f(x) \log f(x) \dx$?

Godavarti and Hero~\cite{GodavartiHero2004}
studied~\eqref{eq:entropy-conv} under a.e.\ convergence and gave
two sufficient conditions:
(i)~$f_n \leq C$ and $\int \abs{x}^\beta f_n \leq C$;
(ii)~a fixed-$\alpha$ condition
\begin{equation}\label{eq:GH-alpha}
  \sup_n \int_{\Rd} f_n(x)\,
  \bigl(\abs{\log f_n(x)}\bigr)^\alpha \dx < \infty,
  \quad \alpha > 1.
\end{equation}
They asked whether the fixed $\alpha > 1$
in~\eqref{eq:GH-alpha} can be replaced by $\alpha_n \downarrow 1$.
To our knowledge, this question has not been resolved in the
subsequent literature.

Piera and Parada~\cite{PieraParada2009} proved entropy convergence
under a.e.\ convergence with bounded density
ratio~$\esssup\, f_n/f \leq C$.
Ghourchian, Gohari, and Amini~\cite{GhourchianGohariAmini2017}
proved entropy continuity under total variation convergence
($\norm{f_n - f}_1 \to 0$) on classes with bounded density and
moments, and gave an explicit modulus of continuity;
we recover their qualitative convergence but not their rates.
For a broader survey, see~\cite{Verdu2019}.

The analytical tools we use---Vitali's theorem, uniform
integrability, de~la Vall\'{e}e--Poussin---are classical, and
Orlicz conditions for integrability are well studied in
real analysis.
Our contribution is an information-theory facing synthesis:
to our knowledge, no prior work---in the information theory
or probability literature---has formulated uniform integrability
and tightness of the entropy integrands
$\{f_n\abs{\log f_n}\}$ as the single mechanism underlying the
sufficient conditions
of~\cite{GodavartiHero2004, PieraParada2009,
GhourchianGohariAmini2017}, nor
stated a probability-weighted Orlicz condition that targets the
information random variable $\abs{\log f_n(X)}$ directly.

\subsection*{Contributions}

\begin{enumerate}
\item \textbf{Organizing principle}
  (Section~\ref{sec:prelim}):
  We show that entropy converges whenever the entropy integrands
  $g_n = f_n\abs{\log f_n}$ are uniformly integrable and tight
  (Lemma~\ref{thm:NS}).  On bounded domains this is also
  necessary: $\Hent(f_n) \to \Hent(f)$ if and only if
  $\{f_n\abs{\log f_n}\}$ is uniformly integrable
  (Theorem~\ref{prop:compact-iff}).

\item \textbf{Orlicz condition}
  (Section~\ref{sec:orlicz}):
  We give a sufficient condition
  $\sup_n \int f_n\, \Psi(\abs{\log f_n}) < \infty$ for a single
  superlinear~$\Psi$ (Theorem~\ref{thm:orlicz}), strictly weaker
  than~\eqref{eq:GH-alpha} for any fixed $\alpha > 1$.

\item \textbf{Disproof of the Godavarti--Hero (GH) conjecture}
  (Section~\ref{sec:counter}):
  We construct an explicit counterexample (Theorem~\ref{thm:counter}).

\item \textbf{Unification}
  (Section~\ref{sec:corollaries}):
  We recover the conditions of~\cite{GodavartiHero2004},
  \cite{PieraParada2009},
  and~\cite{GhourchianGohariAmini2017} as corollaries.
\end{enumerate}

\section{Preliminaries}\label{sec:prelim}

\subsection{Notation}

We work on $(\Rd, \mathcal{B}(\Rd), \lambda)$ with Lebesgue
measure~$\lambda$.  A \emph{pdf} is a measurable
$f : \Rd \to [0,\infty)$ with $\int f\dx = 1$.  The differential
entropy is $\Hent(f) = -\int_{\Rd} f \log f \dx$ ($0\log 0 = 0$,
since $t\log t \to 0$ as $t \downarrow 0$; $\log = \ln$).
All results assume $\Hent(f)$ is finite: the Vitali
framework requires $f\log f \in L^1$, so
$\abs{\Hent(f)} = \infty$ would need a different approach.
Throughout, $f_n \to f$ denotes convergence in
Lebesgue measure.  Both almost everywhere (a.e.) convergence and total
variation (TV) convergence ($\norm{f_n - f}_1 \to 0$) imply
convergence in measure.  Conversely, for nonneg\-ative
integrable functions, convergence in measure plus convergence
of integrals ($\int g_n \to \int g$) implies
$\norm{g_n - g}_1 \to 0$
(see, e.g., {\cite[p.~97]{Folland1999}}).  For pdfs the
second hypothesis holds automatically ($\int f_n = \int f = 1$),
so convergence in measure implies $\norm{f_n - f}_1 \to 0$
(Lemma~\ref{lem:density-UIT}).
Thus our baseline assumption is
effectively equivalent to TV convergence; we state results
under convergence in measure because it is the natural
hypothesis for Vitali's theorem (Theorem~\ref{thm:Vitali}).
We write UI\,\&\,T for
\emph{uniformly integrable and tight}
(Definition~\ref{def:UI}); on finite measure spaces~(C) is
automatic and UI\,\&\,T reduces to classical UI.

\subsection{Uniform integrability and tightness}

\begin{definition}[UI\,\&\,T]\label{def:UI}
A family $\{g_n\}$ of integrable functions on
$(\Omega, \mathcal{F}, \mu)$ is \emph{uniformly integrable
and tight} (UI\,\&\,T) if:
\begin{enumerate}
\item[(A)] \textbf{Uniform $L^1$ bound:}
  $\sup_{n} \int \abs{g_n} \, d\mu < \infty$.
\item[(B)] \textbf{Equi-integrability:}
  For every $\varepsilon > 0$ there exists $M > 0$ such that
  $\sup_n \int_{\{\abs{g_n} > M\}} \abs{g_n} \, d\mu < \varepsilon$.
\item[(C)] \textbf{Tail concentration:}
  For every $\varepsilon > 0$ there exists $K \subset \Omega$
  with $\mu(K) < \infty$ such that
  $\sup_n \int_{\Omega \setminus K} \abs{g_n} \, d\mu < \varepsilon$.
\end{enumerate}
On a finite measure space, (A)+(B) alone---classical uniform
integrability---suffice for Vitali's theorem.
On $\sigma$-finite spaces such as~$(\Rd,\text{Lebesgue})$,
condition~(C) is the additional ingredient that makes Vitali work
(Theorem~\ref{thm:Vitali}).
When $g_n = f_n$, condition~(C) coincides with probability tightness
of $f_n\dx$; but for $g_n = f_n\abs{\log f_n}$, it is a stronger
requirement on the entropy integrands, not merely on the measures.
On a bounded domain (C) is automatic and
UI\,\&\,T reduces to classical UI.
Example~\ref{ex:converse-fails} shows that on~$\Rd$ condition~(C)
can fail for $\{f_n\abs{\log f_n}\}$ even when
$\Hent(f_n) \to \Hent(f)$.
\end{definition}

\subsection{Classical tools}

\begin{theorem}[de~la Vall\'{e}e--Poussin~\cite{delaValleePoussin1915}]
  \label{thm:dVP}
A family $\{g_n\}$ on a finite measure space is uniformly integrable
if and only if there exists a convex $\Phi : [0,\infty) \to [0,\infty)$
with $\Phi(0) = 0$ and $\Phi(t)/t \to \infty$ such that
$\sup_n \int \Phi(\abs{g_n}) \, d\mu < \infty$.
\end{theorem}

\begin{theorem}[Vitali; see {\cite[Thm.~4.5.4]{Folland1999}}]
  \label{thm:Vitali}
Let $(\Omega, \mathcal{F}, \mu)$ be $\sigma$-finite,
$g_n, g \in L^1(\mu)$, $g_n \to g$ in measure.  Then
$\int \abs{g_n - g} \, d\mu \to 0$ if and only if $\{g_n\}$ is
UI\,\&\,T.
\end{theorem}

\subsection{Densities under convergence in measure}

\begin{lemma}\label{lem:density-UIT}
Let $\{f_n\}$ and $f$ be pdfs on~$\Rd$ with $f_n \to f$ in measure.
Then $\norm{f_n - f}_1 \to 0$ and $\{f_n\}$ is UI\,\&\,T.
\end{lemma}

\begin{proof}
Since $\min$ is continuous, $\min(f_n, f) \to f$ in measure,
and $0 \leq \min(f_n, f) \leq f \in L^1$.  We claim
$\int \min(f_n, f) \to \int f$.  If not, some subsequence
$\{n_k\}$ has $\abs{\int \min(f_{n_k}, f) - \int f} \geq \delta
> 0$.  Since $f_{n_k} \to f$ in measure, a further subsequence
converges a.e., and ordinary dominated convergence gives
$\int \min(f_{n_{k_j}}, f) \to \int f$, a contradiction.
Therefore
\begin{align*}
  \int \abs{f_n - f}
  &= \int f_n + \int f - 2\int \min(f_n, f) \\
  &= 2 - 2\int \min(f_n, f) \to 0.
\end{align*}

\medskip\noindent
\textbf{(A):}\; $\sup_n \int f_n = 1$.

\medskip\noindent
\textbf{(B):}\; Fix $\varepsilon > 0$.  Choose $n_0$ with
$\int \abs{f_n - f} < \varepsilon/3$ for $n \geq n_0$, and
$\delta > 0$ with $\lambda(A) < \delta \Rightarrow
\int_A f < \varepsilon/3$.  For $n \geq n_0$ and $M \geq 1/\delta$:
$\lambda(\{f_n > M\}) \leq 1/M \leq \delta$ by Markov, so
$\int_{\{f_n > M\}} f_n
\leq \int \abs{f_n - f} + \int_{\{f_n > M\}} f
< 2\varepsilon/3$.
For $n < n_0$: each $f_n \in L^1$, so choose $M_n$ with
$\int_{\{f_n > M_n\}} f_n < \varepsilon$.  Set
$M = \max(1/\delta, M_1, \ldots, M_{n_0-1})$.

\medskip\noindent
\textbf{(C):}\; Analogous: for $n \geq n_0$,
$\int_{\abs{x}>R} f_n
\leq \int \abs{f_n - f} + \int_{\abs{x}>R} f$.
Choose $R$ for the tail of~$f$, then handle $n < n_0$ individually.
\end{proof}

This extends Scheff\'{e}'s lemma (which assumes a.e.\ convergence)
to convergence in measure via a subsequence argument.

\subsection{Entropy convergence via Vitali's theorem}

\begin{lemma}[Entropy convergence via UI\,\&\,T]\label{thm:NS}
Let $\{f_n\}$ and $f$ be pdfs on $\Rd$ with $f_n \to f$ in measure
and $\Hent(f)$ finite.  If the entropy integrands
$g_n = f_n \abs{\log f_n}$ are UI\,\&\,T, then
$\Hent(f_n) \to \Hent(f)$.
\end{lemma}

\begin{proof}
Set $h_n = f_n \log f_n$, $h = f \log f$.  Since $\Hent(f)$ is
finite, $h \in L^1$.  Since $t \log t$ is continuous on
$[0,\infty)$ (with the convention $0\log 0 = 0$),
$h_n \to h$ in measure.  Since $\{\abs{h_n}\} = \{g_n\}$
is UI\,\&\,T, Theorem~\ref{thm:Vitali} gives
$\int \abs{h_n - h} \to 0$, hence $\Hent(f_n) \to \Hent(f)$.
\end{proof}

\begin{example}[Converse fails on $\Rd$]\label{ex:converse-fails}
Let $f = \1_{[0,1]}$ on~$\R$, so $\Hent(f) = 0$.  For $n \geq 3$,
\[
  f_n(x) = \begin{cases}
    1 - 2/n, & x \in [0,1], \\
    e^n, & x \in [n,\, n + (e^{-n})/n], \\
    e^{-n}, & x \in [2n,\, 2n + (e^n)/n], \\
    0, & \text{otherwise}.
  \end{cases}
\]
Then $\int f_n = 1$ and $f_n \to f$ a.e.  The tall spike's entropy
contribution is $-e^n \cdot n \cdot (e^{-n})/n = -1$, the short
spike's is $-e^{-n} \cdot (-n) \cdot (e^n)/n = +1$; these cancel,
so $\Hent(f_n) \to 0 = \Hent(f)$.
However, each spike contributes exactly~$1$
to $\int f_n \abs{\log f_n}$, and the short spike's support recedes
to infinity, so condition~(C) fails for $\{f_n\abs{\log f_n}\}$.
\end{example}

\begin{theorem}[Bounded domains]\label{prop:compact-iff}
Let $\Omega \subset \Rd$ be bounded, $\{f_n\}$ and $f$ pdfs
on~$\Omega$ with $f_n \to f$ in measure and $\Hent(f)$ finite.  Then
$\Hent(f_n) \to \Hent(f)$ if and only if
$\{f_n \abs{\log f_n}\}$ is uniformly integrable.
\end{theorem}

\begin{proof}
Tightness is automatic on bounded domains.

$(\Longleftarrow)$\;
Lemma~\ref{thm:NS}.

$(\Longrightarrow)$\;
Let $h_n = f_n \log f_n = h_n^+ - h_n^-$ and $h = f\log f =
h^+ - h^-$.  Since $0 \leq h_n^- \leq 1/e$ and $\Omega$ is
bounded, dominated convergence (via the subsequence argument of
Lemma~\ref{lem:density-UIT}) gives $\int h_n^- \to \int h^-$.
Combined with $\Hent(f_n) \to \Hent(f)$, i.e.,
$\int h_n \to \int h$, we get $\int h_n^+ \to \int h^+$.
Now $h_n^+ \to h^+$ in measure (by continuity of
$(\cdot)^+$ and $t\log t$), and
$\min(h_n^+, h^+) \to h^+$ in measure with
$\min(h_n^+, h^+) \leq h^+ \in L^1$; the same subsequence
argument gives $\int \min(h_n^+, h^+) \to \int h^+$.  Hence
$\int |h_n^+ - h^+| = \int h_n^+ + \int h^+ -
2\int\min(h_n^+, h^+) \to 0$.
Similarly $\int |h_n^- - h^-| \to 0$, so
$\int |h_n - h| \to 0$ and the converse of
Theorem~\ref{thm:Vitali} gives UI.
\end{proof}

\section{The Orlicz Condition}\label{sec:orlicz}

\begin{theorem}[Orlicz-type sufficient condition]\label{thm:orlicz}
Let $\{f_n\}$ and $f$ be pdfs on $\Rd$ with $f_n \to f$ in measure
and $\Hent(f)$ finite.  Suppose there exists a convex, nondecreasing
$\Psi : [0,\infty) \to [0,\infty)$ with $\Psi(0) = 0$ and
$\Psi(t)/t \to \infty$
such that
\begin{equation}\label{eq:Psi-bound}
  \sup_n \int_{\Rd} f_n(x)\, \Psi\!\bigl(\abs{\log f_n(x)}\bigr) \dx
  < \infty.
\end{equation}
Then each $\Hent(f_n)$ is finite and $\Hent(f_n) \to \Hent(f)$.
\end{theorem}

\begin{proof}
The idea is: the Orlicz bound~\eqref{eq:Psi-bound} controls large
values of $\abs{\log f_n}$ (de~la Vall\'{e}e--Poussin style),
giving equi-integrability~(B) for $g_n = f_n\abs{\log f_n}$.
Tightness~(C) of $\{g_n\}$ follows from tightness of $\{f_n\}$
(Lemma~\ref{lem:density-UIT}), using the split
$\{\abs{\log f_n} > T\}$ versus
$\{\abs{\log f_n} \leq T\}$ where $g_n \leq T f_n$.

Since $\Psi(t)/t \to \infty$, there exists $T_0$ with
$\Psi(t) \geq t$ for $t \geq T_0$, so
$\int f_n \abs{\log f_n} \leq T_0 + C < \infty$ where
$C = \sup_n \int f_n \Psi(\abs{\log f_n})$; in particular each
$\Hent(f_n)$ is finite.
We verify (A)--(C) of Lemma~\ref{thm:NS} for
$g_n = f_n \abs{\log f_n}$.
Define $\varphi(T) = \sup_{t \geq T} t / \Psi(t) \to 0$.  On
$\{\abs{\log f_n} > T\}$,
$\abs{\log f_n} \leq \varphi(T)\Psi(\abs{\log f_n})$, giving
\begin{equation}\label{eq:Psi-tail}
  \int_{\{\abs{\log f_n} > T\}} g_n \leq \varphi(T)\, C.
\end{equation}
By Lemma~\ref{lem:density-UIT}, $\{f_n\}$ is UI\,\&\,T.

\medskip\noindent
\textbf{(A):}\; $\int g_n \leq T_0 + C$ (established above).

\medskip\noindent
\textbf{(B):}\; $\{g_n > M\} \subset
\{\abs{\log f_n} > T\} \cup
\{\abs{\log f_n} \leq T,\; g_n > M\}$.
The first part contributes $\leq \varphi(T) C$.  On the second,
$g_n \leq T f_n$ and $g_n > M$ forces $f_n > M/T$, so
$\int_{\{\abs{\log f_n} \leq T,\, g_n > M\}} g_n
\leq T \int_{\{f_n > M/T\}} f_n$.
Since $\{f_n\}$ is UI: first choose $T$ so that
$\varphi(T) C < \varepsilon/2$; then, for this fixed~$T$,
choose $M$ so that
$\sup_n \int_{\{f_n > M/T\}} f_n < \varepsilon/(2T)$.

\medskip\noindent
\textbf{(C):}\;
$\int_{\abs{x}>R} g_n
\leq \varphi(T) C + T \sup_n \int_{\abs{x}>R} f_n$.
First choose $T$ so that $\varphi(T) C < \varepsilon/2$;
then, for this fixed~$T$, choose $R$ so that
$\sup_n \int_{\abs{x}>R} f_n < \varepsilon/(2T)$
(possible by tightness of $\{f_n\}$).

\medskip
By Lemma~\ref{thm:NS}, $\Hent(f_n) \to \Hent(f)$.
\end{proof}

If $X_n \sim f_n$, the information random variable is
$\iota_n = -\log f_n(X_n)$, and $\Hent(f_n) = \E[\iota_n]$.
Condition~\eqref{eq:Psi-bound} says
$\sup_n \E_{f_n}[\Psi(\abs{\iota_n})] < \infty$: the tails of the
information random variable are uniformly controlled by a single
superlinear moment.  By de~la Vall\'{e}e--Poussin, this is
equivalent to uniform integrability of
$\{\abs{\iota_n}\}_{n \geq 1}$ under the respective laws of
$X_n \sim f_n$.

\begin{example}[Strictly weaker than any fixed $\alpha$]
  \label{ex:orlicz-examples}
Common choices satisfying $\Psi(t)/t \to \infty$ that are strictly
weaker than $t^\alpha$ for every $\alpha > 1$:
$\Psi(t) = t \log(1{+}t)$,\; $t \log(e{+}t)$,\;
$t (\log(e{+}t))^p$,\; $t \log(1{+}\log(1{+}t))$.

Let $f \equiv 1$ on $[0,1]$ and
\[
  f_n(x) = \begin{cases}
    e^n, & 0 \leq x \leq \delta_n, \\
    c_n, & \delta_n < x \leq 1,
  \end{cases}
\]
where $\delta_n = 1/(n\, e^n \log(1{+}n))$ and $c_n$ normalizes.
With $\Psi(t) = t\log(1{+}t)$, the spike's $\Psi$-moment is
$e^n \cdot n \cdot \log(1{+}n) \cdot \delta_n = 1$,
so $\sup_n \int f_n\, \Psi(\abs{\log f_n}) < \infty$.
But for any fixed $\alpha > 1$, the spike's $\alpha$-moment is
$n^{\alpha-1}/\log(1{+}n) \to \infty$,
so~\eqref{eq:GH-alpha} fails.
The spike height $e^n$ also prevents the conditions of
Corollaries~\ref{cor:PP} and~\ref{cor:GGA}:
the densities are unbounded, and $f_n/f = e^n$ on $[0,\delta_n]$.
Yet Theorem~\ref{thm:orlicz} gives
$\Hent(f_n) \to \Hent(f)$.

\end{example}

\begin{table}[ht]
\centering
\caption{Common Orlicz test functions.}\label{tab:Psi}
\begin{tabular}{@{}lll@{}}
\toprule
$\Psi(t)$ & Growth rate & Recovers \\
\midrule
$t^\alpha$, $\alpha > 1$ & polynomial & GH Thm~4 \\
$t\log(e+t)$ & barely superlinear & (new) \\
$t\log(1+t)$ & barely superlinear & (new) \\
$t\log(1+\log(1+t))$ & very slow & (new) \\
\bottomrule
\end{tabular}
\end{table}

\section{Unification of Existing Results}\label{sec:corollaries}

We recover the sufficient conditions
of~\cite{GodavartiHero2004}, \cite{PieraParada2009},
and~\cite{GhourchianGohariAmini2017} as corollaries.
Each assumes $\Hent(f)$ finite and $f_n \to f$ in measure
(which holds under a.e.\ or TV convergence).
Each of these conditions implies the Orlicz condition
(Theorem~\ref{thm:orlicz}), so all are genuine corollaries of a
single result.  The fixed-$\alpha$ condition~\eqref{eq:GH-alpha}
is a special case with $\Psi(t) = t^\alpha$;
the bounded density ratio of~\cite{PieraParada2009} implies the
Orlicz condition via de~la Vall\'{e}e--Poussin
(Corollary~\ref{cor:PP}).
Example~\ref{ex:orlicz-examples} shows the Orlicz condition is
strictly weaker than each of these.
Table~\ref{tab:roadmap} gives an overview.

\begin{table}[ht]
\centering
\caption{Conditions for entropy convergence under
$f_n \to f$ in measure with $\Hent(f)$ finite.
On bounded domains, the first row is also necessary
(Theorem~\ref{prop:compact-iff}).}\label{tab:roadmap}
\small
\begin{tabular}{@{}ll@{}}
\toprule
Hypothesis & Reference \\
\midrule
\multicolumn{2}{@{}l}{\emph{General framework}} \\[4pt]
$\{f_n\abs{\log f_n}\}$ UI\,\&\,T & Lem.~\ref{thm:NS} \\[6pt]
$\sup_n \!\int\! f_n \Psi(\abs{\log f_n}) \!<\! \infty$ & Thm.~\ref{thm:orlicz} \\[4pt]
\midrule
\multicolumn{2}{@{}l}{\emph{Recovered results}} \\[4pt]
$\int\! f_n \abs{\log f_n}^\alpha \!\leq\! C$, $\alpha \!>\! 1$ & Cor.~\ref{cor:GH4}; \cite{GodavartiHero2004} \\[6pt]
$f_n \!\leq\! C_1$, $\int \!\abs{x}^\beta f_n \!\leq\! C_2$ & Cor.~\ref{cor:GH1}; \cite{GodavartiHero2004, GhourchianGohariAmini2017} \\[6pt]
$f_n / f \leq C$ a.e. & Cor.~\ref{cor:PP}; \cite{PieraParada2009} \\[4pt]
\midrule
\multicolumn{2}{@{}l}{\emph{Checkable criteria}} \\[4pt]
$\text{supp}(f_n) \!\subset\! \Omega$, $\lambda(\Omega) \!<\! \infty$ & Thm.~\ref{prop:compact-iff} (iff) \\[6pt]
$f_n \!=\! f_{X_n} \!*\! f_Z$,\; $\norm{f_Z}_\infty \!<\! \infty$, & Cor.~\ref{cor:GH1} \\
\quad $\sup_n \!\int \!\abs{x}^\beta f_{X_n} \!<\! \infty$ & \\[6pt]
$\sup_n \!\int\! f_n\, e^{s\abs{\log f_n}} \!<\! \infty$, & Thm.~\ref{thm:orlicz} \\
\quad $s > 0$ (sub-exponential) & \\[6pt]
$P_{f_n}\!\bigl(\abs{\log f_n} \!>\! t\bigr) \!\leq\! a\,e^{-bt^\gamma}$ & Thm.~\ref{thm:orlicz} \\
\quad uniformly in $n$ & \\[6pt]
$\sup_n \!\int\! f_n\abs{\log f_n}\log(1{+}\abs{\log f_n}) \!<\! \infty$ & Thm.~\ref{thm:orlicz} \\
\quad (barely superlinear) & \\[4pt]
\bottomrule
\end{tabular}
\end{table}

\subsection{Recovering Godavarti--Hero}

\begin{corollary}[cf.\ {\cite[Theorem~4]{GodavartiHero2004}}]
  \label{cor:GH4}
If $\sup_n \int f_n (\abs{\log f_n})^\alpha < \infty$ for a fixed
$\alpha > 1$, then $\Hent(f_n) \to \Hent(f)$.
\end{corollary}

\begin{proof}
Theorem~\ref{thm:orlicz} with $\Psi(t) = t^\alpha$.
\end{proof}

\begin{corollary}[cf.\ {\cite[Theorem~1]{GodavartiHero2004}}]
  \label{cor:GH1}
If $f_n(x) \leq C_1$ and
$\int \abs{x}^\beta f_n(x) \dx \leq C_2$ for all~$n$
(some $C_1, C_2, \beta > 0$), then $\Hent(f_n) \to \Hent(f)$.
\end{corollary}

\begin{proof}
We show $\sup_n \int f_n\, (\abs{\log f_n})^2 < \infty$ and apply
Theorem~\ref{thm:orlicz} with $\Psi(t) = t^2$.
Since $f_n \leq C_1$, the region $\{f_n \geq 1\}$ contributes
$\leq (\log C_1)^2$.
On $\{f_n < 1\}$, split at $f_n = e^{-\abs{x}^\gamma}$ for
$0 < \gamma \leq \beta/2$:
on $\{f_n \geq e^{-\abs{x}^\gamma}\}$,
$f_n (\abs{\log f_n})^2 \leq f_n \abs{x}^{2\gamma}$, bounded
uniformly by the moment condition;
on $\{f_n < e^{-\abs{x}^\gamma}\}$, with $u = -\log f_n$,
$f_n (\abs{\log f_n})^2 = u^2 e^{-u} \leq C_a e^{-a\abs{x}^\gamma}$,
integrable uniformly in~$n$.
\end{proof}

\subsection{Recovering Piera--Parada}

\begin{corollary}[cf.\ {\cite[Theorem~IV.1]{PieraParada2009}}]
  \label{cor:PP}
If $f > 0$ a.e.\ and $\esssup\, f_n/f \leq C$ for all~$n$, then
$\Hent(f_n) \to \Hent(f)$.
\end{corollary}

\begin{proof}
Since $\Hent(f)$ is finite, $\abs{\log f} \in L^1(f\dx)$.
By Theorem~\ref{thm:dVP} on the probability space $(\Rd, f\dx)$,
there exists a convex superlinear~$\Phi$ with
$\int f\, \Phi(\abs{\log f}) < \infty$;
by capping the derivative ($\Phi' \leftarrow \min(\Phi', e^{t/4})$)
we may assume $\Phi(t) \leq 4e^{t/4}$ while preserving convexity
and superlinearity.
Set $\Psi(t) = \Phi(t/2)$.
With $r_n = f_n/f \in [0, C]$, convexity of~$\Phi$ gives
\begin{align*}
  f_n \Psi(\abs{\log f_n})
  &= r_n f\, \Phi\!\bigl(\tfrac{\abs{\log r_n + \log f}}{2}\bigr) \\
  &\leq \tfrac{r_n f}{2}
    \bigl[\Phi(\abs{\log r_n}) + \Phi(\abs{\log f})\bigr].
\end{align*}
Since $\Phi(t) \leq 4 e^{t/4}$,
$r_n \Phi(\abs{\log r_n}) \leq K$ for a constant~$K$
(for $r_n < 1$: $e^{-s} \cdot 4 e^{s/4} = 4 e^{-3s/4} \leq 4$;
for $r_n \in [1,C]$: $\leq C\Phi(\log C)$).
Hence
$\sup_n \int f_n \Psi(\abs{\log f_n})
\leq \tfrac{K}{2} + \tfrac{C}{2} \int f\, \Phi(\abs{\log f})
< \infty$,
and Theorem~\ref{thm:orlicz} applies.
\end{proof}

\subsection{Recovering Ghourchian--Gohari--Amini}

\begin{corollary}[cf.\ {\cite[Theorem~1]{GhourchianGohariAmini2017}}]
  \label{cor:GGA}
If $f_n(x) \leq m$, $\int \abs{x}^\alpha f_n(x) \dx \leq v$ for
all~$n$, and $\norm{f_n - f}_1 \to 0$, then
$\Hent(f_n) \to \Hent(f)$.
\end{corollary}

\begin{proof}
TV convergence implies convergence in measure.
Corollary~\ref{cor:GH1} with $C_1 = m$, $C_2 = v$,
$\beta = \alpha$.
\end{proof}

\begin{remark}[Applications]\label{rem:applications}
The framework developed here applies in any setting where
density convergence must be lifted to entropy convergence---including
plug-in entropy estimation, mutual information continuity along
sequences of channel inputs, and entropy convergence in CLT-type
approximations.  On bounded domains (e.g., bounded-amplitude
signals), Theorem~\ref{prop:compact-iff} gives a complete answer.
On~$\Rd$, Table~\ref{tab:roadmap} lists practical sufficient
conditions that can be verified directly from model parameters.
\end{remark}

\section{The GH Conjecture is False}\label{sec:counter}

Godavarti and Hero~\cite{GodavartiHero2004} asked whether
the fixed $\alpha > 1$ in~\eqref{eq:GH-alpha} can be replaced
by a sequence $\alpha_n \downarrow 1$:

\medskip\noindent
\textbf{GH Conjecture.}\;
\emph{If\/ $f_n \to f$ a.e.\ and\/
$\sup_n \int f_n\, (\abs{\log f_n})^{\alpha_n} < \infty$
for some $\alpha_n \downarrow 1$, then
$\Hent(f_n) \to \Hent(f)$.}

\medskip
We disprove this with an explicit counterexample.

\begin{theorem}[Counterexample]\label{thm:counter}
There exist pdfs $\{f_n\}$ and $f$ on $[0,1]$ such that:
\begin{enumerate}
\item[\textup{(i)}] $f_n \to f$ a.e.\ (hence in measure);
\item[\textup{(ii)}] $\sup_n \int_0^1 f_n
  (\abs{\log f_n})^{\alpha_n} \dx < \infty$
  for $\alpha_n \downarrow 1$;
\item[\textup{(iii)}] $\Hent(f_n) \not\to \Hent(f)$.
\end{enumerate}
\end{theorem}

\begin{proof}
Let $f \equiv 1$ on $[0,1]$ ($\Hent(f) = 0$).  Set
$\delta_n = (e^{-n})/n$,
$\alpha_n = 1 + 1/\log n$, and
\begin{equation}\label{eq:fn-def}
  f_n(x) =
  \begin{cases}
    e^n, & 0 \leq x \leq \delta_n, \\
    c_n, & \delta_n < x \leq 1,
  \end{cases}
\end{equation}
where $c_n = (1 - 1/n)/(1 - (e^{-n})/n)$.
For $n \geq 2$, $(e^{-n})/n < 1$ so $c_n > 0$, and $c_n \to 1$.

\medskip\noindent
\textbf{(i):} For fixed $x \in (0,1]$, $\delta_n \to 0$ so
$f_n(x) = c_n \to 1$ eventually.
Moreover, $\norm{f_n - f}_1 \leq (e^n - 1)\delta_n +
\abs{c_n - 1}(1 - \delta_n) \to 0$.

\medskip\noindent
\textbf{(ii):} The spike's $\alpha_n$-moment is
$e^n \cdot n^{\alpha_n} \cdot (e^{-n})/n = n^{\alpha_n - 1}
= n^{1/\log n} = e$.
The background contribution $\to 0$.  So
$\sup_n \int f_n (\abs{\log f_n})^{\alpha_n} \leq e + 1 < \infty$.

\medskip\noindent
\textbf{(iii):} The spike's entropy contribution is
$-e^n \cdot n \cdot (e^{-n})/n = -1$.
The background $\to 0$.  So $\Hent(f_n) \to -1 \neq 0$.
\end{proof}

The failure mechanism is clear: the moving-exponent condition uses
$\Psi_n(t) = t^{\alpha_n}$ with $\alpha_n \downarrow 1$, but no
single $\Psi$ with $\Psi(t)/t \to \infty$ is dominated by
all~$\Psi_n$, since $\Psi(t)/t \leq t^{\alpha_n - 1} \to 1$ for
any fixed~$t$ as $n \to \infty$.
Since the counterexample lives on~$[0,1]$ where tightness is
automatic, Theorem~\ref{prop:compact-iff} further shows that the
moving-exponent condition does not even imply uniform integrability
of $\{f_n \abs{\log f_n}\}$.
For any fixed $\alpha > 1$, the spike's $\alpha$-moment is
$n^{\alpha - 1} \to \infty$, so~\eqref{eq:GH-alpha} fails
and Theorem~\ref{thm:orlicz} does not apply.

\section{Conclusion}\label{sec:conclusion}

We have shown that Vitali's convergence theorem reduces entropy
convergence to UI\,\&\,T of the entropy integrands
$f_n \abs{\log f_n}$ (Lemma~\ref{thm:NS}).  On bounded domains,
this condition is both necessary and sufficient
(Theorem~\ref{prop:compact-iff}), giving a complete
characterization of entropy convergence.  We gave an
Orlicz condition (Theorem~\ref{thm:orlicz}) that is strictly weaker
than the sufficient conditions
of~\cite{GodavartiHero2004, PieraParada2009,
GhourchianGohariAmini2017}, and disproved the Godavarti--Hero
conjecture (Theorem~\ref{thm:counter}).  All results assume
convergence in Lebesgue measure (equivalently, TV convergence for
pdfs); they do not address entropy continuity under weak
convergence of distributions.
On~$\Rd$, UI\,\&\,T of $\{f_n\abs{\log f_n}\}$ is sufficient but
not necessary (Example~\ref{ex:converse-fails}); characterizing
necessity without tightness remains open.

\bibliographystyle{IEEEtran}
\bibliography{references}

\end{document}